\pdfoutput=1
\RequirePackage{setspace}
\documentclass[a4paper,USenglish,cleveref,autoref,thm-restate]{lipics-v2021}
\usepackage{acronym,tabularx}

\makeatletter
\@ifclassloaded{llncs}{}{}
\makeatother

\makeatletter
\@ifclassloaded{beamer}{\PassOptionsToPackage{most}{tcolorbox}}{}
\makeatother

\PassOptionsToPackage{T1}{fontenc}
\PassOptionsToPackage{toc}{appendix}
\PassOptionsToPackage{utf8}{inputenc}
\PassOptionsToPackage{capitalize}{cleveref}
\PassOptionsToPackage{noend}{algpseudocode}
\PassOptionsToPackage{shortlabels}{enumitem}

\makeatletter
\@ifclassloaded{moderncv}{}{\PassOptionsToPackage{unicode=true}{hyperref}}
\makeatother

\makeatletter
\@ifclassloaded{acmart}{}{\usepackage{amssymb}}
\@ifclassloaded{lipics-v2021}{}{\usepackage{setspace}}
\makeatother

\usepackage{
  graphicx, xcolor, tikz,
  float,
  amsfonts,
  amsthm,
  mathtools,
  bm,
  braket, interval,
  booktabs, algorithm, algpseudocode, multirow,
  inputenc, fontenc,
  csquotes,
  thmtools, thm-restate
}

\makeatletter
\@ifclassloaded{moderncv}{}{
  \@ifclassloaded{IEEEtran}{}{\usepackage{appendix}}
  \@ifclassloaded{beamer}{}{
    \@ifclassloaded{lipics-v2021}{}{
      \usepackage{nameref,varioref,hyperref,cleveref}
    }
  }
}
\makeatother

\makeatletter
\@ifclassloaded{IEEEtran}{}{
  \g@addto@macro\bfseries{\boldmath}
  \g@addto@macro\normalfont{\unboldmath}
}
\@ifclassloaded{beamer}{
  \setbeamercolor{frametitle}{fg=black!70!white,bg=logobg}
  \let\@lgorithmforbe@merfix\algorithm
  \let\end@lgorithmforbe@merfix\endalgorithm
  \renewenvironment{algorithm}[1][H]{\@lgorithmforbe@merfix[#1]}{\end@lgorithmforbe@merfix}
  
  \tcbset{on line,boxsep=5pt,left=0pt,right=0pt,top=0pt,bottom=0pt,frame empty,frame hidden,colback=lightgray,highlight math style={enhanced}}
  \newcommand{\tcb}[2][]{\relax\ifmmode\tcbhighmath[#1]{#2}\else\tcbox[#1]{#2}\fi}%
  
}{}
\makeatother

\usetikzlibrary{decorations.pathreplacing,calligraphy,patterns,arrows.meta,external}

\intervalconfig{scaled,soft open fences}

\makeatletter
\@ifpackageloaded{cleveref}{
\crefname{appsec}{Appendix}{Appendices}
}{}
\newcommand\DeclareMathOperators[1]{%
  \@for\@ii:=#1\do{\expandafter\DeclareMathOperator\@ii}%
}%
\makeatother

\newcommand{\Z}{\mathbb{Z}}
\newcommand{\Q}{\mathbb{Q}}
\newcommand{\R}{\mathbb{R}}
\newcommand{\NP}{\textsc{\textup{NP}}}

\newcommand{\todo}[1]{\noindent\textcolor{purple}{\textsc{todo:} #1}}

\newcommand{\submit}{\renewcommand{\todo}[1]{}}

\makeatletter
\newcommand\useplaintitle{
    \def\@maketitle{
    	\newpage
    	\null
    	\vskip 2em%
    	\begin{center}%
    		\let \footnote \thanks
    		{\LARGE \@title \par}%
    		\vskip 1.5em%
    		{\large
    			\lineskip .5em%
    			\begin{tabular}[t]{c}%
    				\@author
    			\end{tabular}\par}%
    	\end{center}%
    	\par
    	\vskip 1em
    }
}
\newcommand*{\currentname}{\@currentlabelname}
\makeatother
\newcommand{\num}{\addtocounter{equation}{1}\tag{\theequation}}
\newcommand{\reason}[1]{\tag*{\textcolor{gray}{#1}}}

\newcommand{\op}[1]{\operatorname{#1}}

\DeclarePairedDelimiter\ceil{\lceil}{\rceil}
\DeclarePairedDelimiter\floor{\lfloor}{\rfloor}
\DeclarePairedDelimiter\norm{\lVert}{\rVert}

\makeatletter
\@ifclassloaded{acmart}{}{}
\makeatother

\newcommand{\overbar}[1]{\mkern 1.5mu\overline{\mkern-1.5mu#1\mkern-1.5mu}\mkern 1.5mu}

\newcommand{\para}[1]{\noindent\textbf{#1}\,\,}

\DeclareMathOperators{
  \Oh{\mathcal{O}},\oh{o},
  \OPT{OPT},\chp{chp},\load{L},
  \ks{ks},\cks{cks},
  \size{size},\pop{pop},\push{push},\parent{parent},
  \lcm{lcm},\supp{supp},\sign{sign},\poly{poly}
}

\acrodef{RTA}[\textsc{rta}]{Response Time Analysis}
\acrodef{FP}[\textsc{fp}]{fixed-priority}
\acrodef{DM}[\textsc{dm}]{deadline monotonic}
\acrodef{RM}[\textsc{rm}]{rate monotonic}
\acrodef{SAS}[\textsc{sas}]{synchronous arrival sequence}
\usetikzlibrary{calc,math}

\title{The Power of Duality: Response Time Analysis meets Integer Programming}

\author{Max A. Deppert}{Kiel University, Kiel \and \url{https://max.deppert.de}}{made@informatik.uni-kiel.de}{https://orcid.org/0000-0003-3083-7998}{}
\author{Klaus Jansen}{Kiel University, Kiel}{kj@informatik.uni-kiel.de}{https://orcid.org/0000-0001-8358-6796}{}

\authorrunning{M.\,A. Deppert and K. Jansen} 

\Copyright{Max A. Deppert and Klaus Jansen} 



\category{} 

\relatedversion{} 

\funding{Research supported by German Research Foundation (DFG) project JA 612/25-1}

\acknowledgements{We want to thank Werner Grass, whose numerous comments and suggestions have helped us greatly to improve our paper.}

\nolinenumbers 
\hideLIPIcs

\EventEditors{John Q. Open and Joan R. Access}
\EventNoEds{2}
\EventLongTitle{42nd Conference on Very Important Topics (STACS 2024)}
\EventShortTitle{STACS 2024}
\EventAcronym{STACS}
\EventYear{2024}
\EventDate{December 24--27, 2016}
\EventLocation{Little Whinging, United Kingdom}
\EventLogo{}
\SeriesVolume{42}
\ArticleNo{23}

\newcommand\response{\mathscr{R}}
\newcommand\responsedual{\overbar{\mathscr{R}}}
\newcommand\mixing{\mathscr{M}}
\newcommand\jitter{\eta}

\definecolor{TolDarkPurple}{HTML}{332288}
\definecolor{TolLightPurple}{HTML}{AA4499}
\definecolor{TolDarkBlue}{HTML}{6699CC}
\definecolor{TolLightBlue}{HTML}{88CCEE}
\definecolor{TolLightGreen}{HTML}{44AA99}
\definecolor{TolDarkGreen}{HTML}{117733}
\definecolor{TolDarkBrown}{HTML}{999933}
\definecolor{TolLightBrown}{HTML}{DDCC77}
\definecolor{TolDarkRed}{HTML}{661100}
\definecolor{TolLightRed}{HTML}{CC6677}
\definecolor{TolLightPink}{HTML}{AA4466}
\definecolor{TolDarkPink}{HTML}{882255}

\algnewcommand{\IfThenElse}[3]{
  \State \algorithmicif\ #1\ \algorithmicthen\ #2\ \algorithmicelse\ #3}

\submit

\begin{document}

\maketitle

\begin{abstract}
    We study a mutually enriching connection between response time analysis in real-time systems and the mixing set problem.
    Thereby generalizing over known results we present a new approach to the computation of response times in fixed-priority uniprocessor real-time scheduling. We even allow that the tasks are delayed by some period-constrained release jitter. By studying a dual problem formulation of the decision problem as an integer linear program we show that worst-case response times can be computed by algorithmically exploiting a conditional reduction to an instance of the mixing set problem.
    In the important case of harmonic periods our new technique admits a near-quadratic algorithm to the exact computation of worst-case response times. We show that generally, a smaller utilization leads to more efficient algorithms even in fixed-priority scheduling.
    Worst-case response times can be understood as least fixed points to non-trivial fixed point equations and as such, our approach may also be used to solve suitable fixed point problems.
    Furthermore, we show that our technique can be reversed to solve the mixing set problem by computing worst-case response times to associated real-time scheduling task systems.
    Finally, we also apply our optimization technique to solve 4-block integer programs with simple objective functions.
\end{abstract}

\clearpage

\section{Introduction}

\todo{cite 
Bini and Buttazzo \cite{DBLP:journals/tc/BiniB04},
Baruah et al. \cite{DBLP:conf/rtss/BaruahE022}
}

We consider fixed-priority uniprocessor real-time task scheduling with constrained deadlines and task release jitters in the \emph{sporadic task model} \cite{DBLP:phd/ndltd/Mok83} which is a common model to analyze real-time task systems.




\paragraph*{Model}

A sporadic task system $\mathscr{T}=\set{\tau_1,\dots,\tau_n}$ is a set of sporadic tasks $
\tau_i$ which are given as a quadruple of non-negative integers $\tau_i = (c_i,d_i,p_i,\jitter_i)$ where $c_i$ represents the task's worst-case execution time, $d_i$ its relative deadline, $p_i$ its period, and $\jitter_i$ its release jitter. The task $\tau_i$ generates an infinite sequence of jobs. Each job has an execution time of at most $c_i \geq 1$ time units and a deadline exactly $d_i$ time units after its arrival time, where $c_i \leq d_i \leq p_i$. The jobs of the task are considered to arrive separated in time by at least $p_i\geq 1$ time units. However, a job is not ready for execution until it is \emph{released}. The time difference between the arrival time and the release time is called \emph{release jitter}. Hence,
in jitter-free systems the jobs are released in the moment they arrive, i.e.
arrival time equals release time.
The release jitter $\jitter_i$ of the task is the maximum time difference between the arrival times and the release times over all jobs of $\tau_i$. Naturally, one can assume that $\jitter_i \leq d_i-c_i \leq p_i-c_i$.

The task system $\mathscr{T}$ has \emph{harmonic periods} iff $p_i \geq p_j$ implies that $p_i/p_j\in\Z$ for all tasks $\tau_i,\tau_j$.
The \emph{utilization} of a task $\tau_i$ is the quantity $c_i/p_i$ and the utilization of the task set $\mathscr{T}$ is $\sum_{i\leq n}c_i/p_i$. The common assumption is that $\sum_{i\leq n}c_i/p_i \leq 1$ as otherwise there are job sequences of $\mathscr{T}$ which can not be scheduled by any algorithm. We refer to this utilization bound as the \emph{schedulability utilization bound}. However, a necessary utilization bound is $\sum_{i<n}c_i/p_i < 1$ as otherwise there is no solution for worst-case response times at all. We refer to this bound as the \emph{general utilization bound}.

\paragraph*{Fixed-Priority Scheduling}

\begin{figure*}
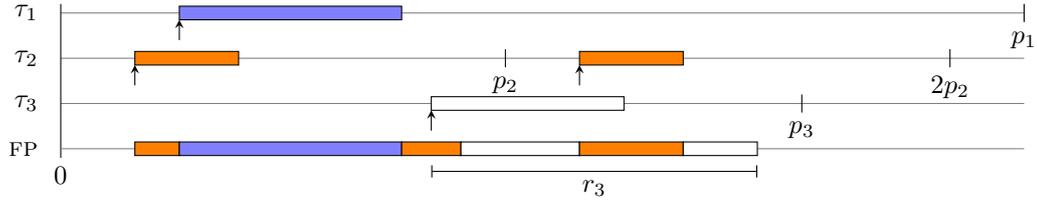

    \centering
    \tikz[xscale=0.195,yscale=0.6]{
        \def\padding{.15}
        \def\ind{.2}
        \foreach [count=\i] \offset/\c/\p/\jit/\periods/\ops in {
            0/15/65/8/1/[fill=white!50!blue],
            0/7/30/5/2/[fill=orange],
            0/13/50/25/1/[fill=white]
        } {
            \draw [gray] (0,-\i) node [left=5,black] {$\tau_{\i}$} -- (65,-\i);
            \foreach \j in {1,...,\periods} {
                \def\periodcoefficient{\ifnum\j>1{\j}\else{}\fi}
                \draw (\offset+\j*\p,-\i+\ind) -- (\offset+\j*\p,-\i-\ind) node [below] {$\periodcoefficient p_{\i}$};
                \draw [-stealth] (\offset+\j*\p-\p+\jit,-\i-4*\padding) -- (\offset+\j*\p-\p+\jit,-\i-\padding);
                \draw \ops (\offset+\j*\p-\p+\jit,-\i-\padding) rectangle (\offset+\j*\p-\p+\jit+\c,-\i+\padding);
            }
        }
        \draw (0,-1+\ind) -- (0,-4-\ind) node [below] {$0$};
        \draw [gray] (0,-4) node [left=5,black] {\textsc{fp}} -- (65,-4);
        \draw [fill=orange] (5,-4-\padding) rectangle (8,-4+\padding);
        \draw [fill=white!50!blue] (8,-4-\padding) rectangle (23,-4+\padding);
        \draw [fill=orange] (23,-4-\padding) rectangle (27,-4+\padding);
        \draw [fill=white] (27,-4-\padding) rectangle (35,-4+\padding);
        \draw [fill=orange] (35,-4-\padding) rectangle (42,-4+\padding);
        \draw [fill=white] (42,-4-\padding) rectangle (47,-4+\padding);
        \draw [|-|] (25,-4.5) -- node [below] {$r_3$} (47,-4.5);
    }
    \caption{An example instance $\mathscr{T} = \set{\tau_1,\tau_2,\tau_3}$ with $\tau_1 = (15,d_1,65,8)$, $\tau_2 = (7,d_2,30,5)$, $\tau_3 = (13,d_3,50,25)$ and the resulting \acs{FP} schedule. The arrows indicate the release times of the jobs and the colors distinguish between jobs of different tasks.}
    \label{fig:example}
\end{figure*}

In \ac{FP} scheduling (or static-priority scheduling) the tasks are considered to have a predefined order of prioritization. A task $\tau_i$ has a larger priority than a task $\tau_j$ if and only if $i<j$. An \ac{FP} schedule preempts jobs according to these priorities; in more detail, a running job is preempted on the release of a job of task of a higher priority to schedule the new job instead (see \cref{fig:example} for an example). The task system is called \emph{schedulable} when no job can ever miss its deadline. Of course one may try to find a good prioritization for a set of tasks. However, the more typical problem is to analyze a task system with a given prioritization which is what we do in this paper.

Two important special cases of \ac{FP} scheduling are \ac{DM} scheduling \cite{DBLP:journals/rts/Baruah11,AUDSLEY1991127} and \ac{RM} scheduling \cite{DBLP:conf/rtss/Ekberg20,DBLP:journals/jacm/LiuL73,DBLP:conf/afips/Serlin72}.
In \ac{DM} scheduling earlier deadlines imply larger priorities while in \ac{RM} scheduling smaller periods imply larger priorities. Both are equal if the deadline of each task is equal to its period.

\paragraph*{Response Time Analysis and Schedulability Testing}

The \emph{response time} of a task is the maximum time difference between the release of a job and its completion. In a worst-case scenario we refer to them as \emph{worst-case response times} and these may be used to decide the schedulability of a system. Since task release jitters intensify the complexity of worst-case response times and schedulability testing as well as other types of analysis, they are often assumed to be zero, i.e. $\jitter_i=0$ for all $i\leq n$.
As pointed out by Liu and Layland \cite{DBLP:journals/jacm/LiuL73}, in such \emph{jitter-free} systems the \emph{critical instant} occurs when the jobs of all tasks $\tau_i$ are released simultaneously. See also \ac{SAS} \cite{DBLP:reference/crc/BaruahG04,DBLP:journals/cj/JosephP86}.
Therefore, without release jitters the worst-case response time of task $\tau_j$ is the earliest point in time $r_j$ where
\[
    \textstyle
    r_j \geq c_j + \sum_{i=1}^{j-1}c_i\ceil{r_j/p_i}.
\]
This inequality models the situation that all jobs are released at time $0$.
If $r_j \leq d_j$, then task $t_j$ completes its job on time which means that task $\tau_j$ is schedulable. If $r_j \leq d_j$ for all tasks $t_j$, then the task system $\mathscr{T}$ is schedulable using the given priorities.


In this paper we consider systems with release jitters. Audsley et al. \cite{DBLP:journals/iee/AudsleyBRTW93} and Tindell et al. \cite{DBLP:journals/rts/TindellBW94} have discussed the release jitter problem in detail (see also \cite{DBLP:conf/rtss/SjodinH98,DBLP:conf/rtss/BiniPD15,DBLP:conf/rtss/ChenHL16,DBLP:journals/rts/GrassN18,DBLP:conf/ecrts/BoyerRDP21}). Their work leads to the following understanding of worst-case response times in systems with task release jitters.
The worst-case response time of task $\tau_j$ is the earliest point in time $r_j$ where
\[
    r_j \geq c_j + \sum_{i=1}^{j-1}c_i\ceil{(r_j+\jitter_i)/p_i}.
\]
The task is schedulable if $r_j \leq d_j-\jitter_j$ and the task set $\mathscr{T}$ is schedulable if $r_j \leq d_j-\jitter_j$ for all tasks $\tau_j$.
Without loss of generality we can assume that the execution requirement of each job of a task $\tau_i$ always equals the task's worst-case execution time $c_i$. Also we can assume that the $i$-th job of a task $\tau_j$ is released in the time interval $[(i-1)p_j,(i-1)p_j+\jitter_j]$ and its due time is $(i-1)p_j+d_j$.
Furthermore, we may restrict our interest to the response time $r_n$ of the last task $\tau_n$:
\[\num\label{wcrt-with-jitters-original}
    r_n = \min\set{t|t \geq c_n + \sum_{i<n} c_i\ceil[\Big]{\frac{t+\jitter_i}{p_i}}, t \in \Z_{\geq 0}}.
\]
We refer to the problem of the computation of $r_n$ as \emph{response time computation} (RTC).

For \ac{FP} scheduling, RTC is NP-hard in general due to Eisenbrand and Rothvoß \cite{DBLP:conf/rtss/EisenbrandR08} even for $\jitter=0$ and it is even hard to approximate within a constant factor. However, for jitter-free task systems with harmonic periods it is polynomial which is a great result by Bonifaci~et~al.~\cite{DBLP:conf/rtss/BonifaciMMW13} and it may even be computed in near-linear time, cf. Nguyen~et~al.~\cite{DBLP:conf/iwoca/NguyenGJ22}.

The computation of response times in real-time systems is often treated as a fixed point problem, where one aims for the least fixed point to the fixed point equation $\Phi(t) = t$ where $\Phi(t) = c_n+\sum_{i<n}c_i\ceil{(t+\jitter_i)/p_i}$. The computation of least fixed points (e.g. \cite{DBLP:books/daglib/0095988,DBLP:books/sp/Libkin04}) is an important topic in order theory and logic with many applications, for example in denotational semantics. As such, the new approach which we present in this paper may also be used to find least fixed points of suitable fixed point problems .

\paragraph*{Mixing Set}

In the \textsc{Mixing Set} problem one is given capacities $a\in\Q^n$ and a right-hand side $b\in\Q^n$ and the goal is to find a solution $(s,x)\in\R_{\geq 0}\times\Z^n$ which optimizes a linear objective function while satisfying the following system of inequalities:
\[
    s+a_ix_i \geq b_i \quad \forall i\in\set{1,\dots,n}
\]
The problem is \NP-hard due to Eisenbrand and Rothvoß \cite{DBLP:conf/approx/EisenbrandR09}. However, it can be solved in polynomial time in the case of unit capacities \cite{DBLP:journals/mp/GunlukP01,DBLP:journals/mp/MillerW03} or harmonic capacities \cite{DBLP:journals/mp/ZhaoF08} (see also \cite{DBLP:conf/ipco/ConfortiSW08,DBLP:journals/orl/ConfortiZ09} for simpler approaches). The \textsc{Mixing Set} problem plays an important role in production planning (in particular lot-sizing \cite{Pochet:2006:PPM:1202598}).

\paragraph*{Related Work}

In the early '70s it was independently shown by Liu and Layland \cite{DBLP:journals/jacm/LiuL73} and Serlin \cite{DBLP:conf/afips/Serlin72}, that a set of $n$ implicit-deadline (i.e. $d_i=p_i$) synchronous periodic tasks is always schedulable in \ac{RM} scheduling if the task's total utilization is bounded by $n(2^{1/n}-1)$ which is $\ln(2)\approx 0.69$ for $n\rightarrow\infty$. Recently, this seminal result was complemented by Ekberg \cite{DBLP:conf/rtss/Ekberg20} who proved that $\ln(2)$ is in fact the boundary between polynomial and NP-hard schedulability testing; the \ac{FP} schedulability problem is NP-hard even if it is restricted to task sets with \ac{RM} priorities and utilization-bounded from above by any constant truly larger than $\ln(2)$.

The Liu and Layland problem was investigated further for example by Kubiak \cite{DBLP:journals/scheduling/Kubiak05} who solved it via bottleneck just-in-time sequencing and Józefowska et al. \cite{DBLP:journals/eor/JozefowskaJK09} who studied apportionment divisor methods.

Baruah \cite{DBLP:conf/rtss/Baruah20} introduced the idea of ILP-tractability of schedulability analysis problems and in fact, our results prove tractability in his sense. \todo{specify}

As mentioned earlier, Eisenbrand and Rothvoß proved hardness results for both RTC \cite{DBLP:conf/rtss/EisenbrandR08} and the \textsc{Mixing Set} problem \cite{DBLP:conf/approx/EisenbrandR09}; in more detail, they proved reductions from directed diophantine approximation \emph{to} these problems but they did not give any reductions from them or \emph{between} them, which is what we have found.

\begin{table*}
    \centering
    \begin{tabularx}{\linewidth}{p{3cm}|X|p{4cm}}
        \toprule
            \textsc{RTC} & harmonic periods & arbitrary periods
        \\\hline
            $\jitter_i \leq p_i$
            &
            \textbf{$\Oh(n^2\log p_{\max})$} *
            &
            $\Oh(n\lcm_{i<n}p_i)$ *\newline
            $\Oh(\max(Sn,T_{\mixing}\log p_{\max}))$ *
        \\\hline
            $\jitter_i=0$
            &
            $\Oh(n\log(n+p_{\max}))$ \cite{DBLP:conf/rtss/BonifaciMMW13}, $\Oh(n\log n)$ \cite{DBLP:conf/iwoca/NguyenGJ22}
            &
            $\Oh(T_{\mixing}\log p_{\max})$ *
        \\\midrule
            \textsc{Mixing Set} & harmonic capacities & arbitrary capacities
        \\\hline
            $b_i \in \Z$
            &
            $\Oh(n^2)$ \cite{DBLP:journals/orl/ConfortiZ09}
            &
            $\Oh(n\lcm_{i\leq n}a_i)$ *\newline
            $\Oh(T_{\response}\log b_{\max})$ *
        \\\hline
            $b_i = \beta \geq \lcm_{j\leq n} a_j$
            &
            $\Oh(n\log(n+\beta))$ *
            &
            $\Oh(T_{\response}\log\beta)$ *
        \\\bottomrule
    \end{tabularx}
    \medskip
    \caption{\normalshape Overview of algorithmic results for RTC and \textsc{Mixing Set}. Here $S < (\sum_{i<n}c_i)/(1-U)$ for utilization $U=\sum_{i<n}c_i/p_i$ is an upper bound on the value of $s$ in any optimal solution for a certain instance of \textsc{Mixing Set}, and $T_{\mixing}$ and $T_{\response}$ are upper bounds on the time required to solve an instance of \textsc{Mixing Set} or RTC, respectively.\hfill * Result is in this paper}
    \label{tab:results-overview}
\end{table*}

\paragraph*{Our Contribution}
Here we summarize our results.
\begin{enumerate}[(a)]
    \item\label{dualision-teaser}
    We present a simple dualization technique which allows to solve inequality-constrained optimization problems in a binary search by consecutively solving a dual formulation of the associated decision problem (see \cref{dualision}).
    \item \para{Response Times}
    \begin{enumerate}[1.]
        \item We establish new bounds for worst-case response times in \ac{FP} real-time task systems with period-constrained task release jitter (see \cref{sec:response-time-bounds}).
        \item We present a conditional Karp reduction from RTC to the \textsc{Mixing Set} problem (see \cref{conditional-karp-reduction}). We show how it can be used to derive a bunch of algorithmic results for both problems (see \cref{tab:results-overview} for an overview).
        \item For the important case of harmonic periods we prove a Cook reduction from RTC to \textsc{Mixing Set} (see \cref{wcrt-harmonic-periods}).
        \item We can even cope with the general case of \emph{arbitrary} periods. Especially, we give a Turing reduction from RTC to \textsc{Mixing Set} and we study the dependence of the running times on the utilization (see \cref{wcrt-arbitrary-periods}).
    \end{enumerate}
    \item \para{Mixing Set} We can reverse our methodology to solve the \textsc{Mixing Set} problem by computing worst-case response times for associated real-time task systems (see \cref{sec:mixing-set}).
    \item \para{4-block Programs} Finally, we show how the dualization technique can by applied to solve certain 4-block integer programs (see \cref{application-to-4-block}) which are NP-hard to approximate to any constant factor (see \cref{4-block-hardness}).
\end{enumerate}


\paragraph*{RTC for Task Systems with Harmonic Periods}




We want to highlight our result for harmonic periods, as it gives an efficient algorithm to a well-studied case of a long-standing problem. On the one hand, the release jitter problem was studied for about 30 years now (e.g.
\cite{
    DBLP:journals/iee/AudsleyBRTW93,%
    DBLP:journals/rts/TindellBW94,%
    DBLP:conf/rtss/BaruahCM97,%
    DBLP:conf/rtss/DavisB08,%
    DBLP:conf/rtss/BiniPD15%
}), and also harmonic periods have been studied (e.g.
\cite{
    DBLP:journals/tc/KuoM97,%
    DBLP:conf/rtss/BonifaciMMW13,%
    d7c4f1a1-dce5-4ef9-a888-3089d808cec5,%
    DBLP:conf/iwoca/NguyenGJ22%
}) but on the other hand and to the best of our knowledge, there was no polynomial-time algorithm known before. Also remark that \ac{FP} scheduling superceeds the generality of \ac{DM} and \ac{RM} scheduling and in fact, we never rely on their properties. In the following we give a brief overview about our algorithm for harmonic periods.

\begin{theorem}\label{main-result}
    For systems with task release jitters and harmonic periods, the worst-case response time $r_n$ can be computed in time $\Oh(n^2\log p_{\max})$.
\end{theorem}

We consider the decision problem $r_n \leq_? k$ for which we present a conditional Karp reduction to the \textsc{Mixing Set} problem in constant time. The reduction is conditional in the sense that it works for large values of $k$ only, i.e. $k \geq S$ for some value $S$ which bounds the value of the variable $s$ in any optimal solution to the arising instances of \textsc{Mixing Set}. However, we show that $S$ can be bounded by the worst-case response time $r_n'$ for the simplified instance without release jitters and as this response time itself is a lower bound of $r_n$, i.e. $r_n' \leq r_n$, this suffices to run a binary search for $r_n$ by solving instances of \textsc{Mixing Set}.
Together, this yields a Cook reduction to \textsc{Mixing Set}.

\paragraph*{Notation}

We define $[n]=\set{z\in\Z|1\leq z \leq n}$ to denote the set of the first $n$ natural numbers. Given a vector $v=(v_1,\dots,v_n)\in\R^n$ we denote the largest or smallest entry in $v$ by $v_{\max} = \max_{i\leq n}v_i$ or $v_{\min} = \min_{i\leq n}v_i$, respectively. For positive integers $z_1,\dots,z_n$ let $\lcm_{i\leq n} z_i$ denote the least common multiple of the numbers $z_1,\dots,z_n$. For sizes $a,b$ we write $a \leq_? b$ to denote the decision problem to decide whether $a \leq b$ is true or not, i.e. the correct answer to $a \leq_? b$ is \textsc{yes} if $a \leq b$ and \textsc{no} otherwise. In the same manner we define $a \geq_? b$.

Whenever we specify the computation time of an algorithm in $\Oh$-notation, we quantify its total number of arithmetic operations on numbers from the input. With $\log=\log_2$ we always refer to the \emph{binary} logarithm.



\section{Preliminaries}

Here we introduce a dualization technique for optimization problems. Also we show first insights about the \textsc{Mixing Set} problem and present some handy bounds for response times in real-time systems.

\subsection{A Dualization Technique}\label{dualision}

We present a simple optimization technique by studying a dual formulation of the associated decision problem. Consider a general optimization problem as follows.
Let $X$ be a set and let $f:X\rightarrow\Z$ and $g:X\rightarrow\R$ be two computable functions and we seek to minimize $f(x)$ over all $x\in X$ while $g(x) \geq b$ for some right-hand side $b\in\R$ and other constraints $G(x)$ have to be satisfied, i.e. we aim to compute
\[
    \sigma = \min\set{f(x)|g(x)\geq b,\;G(x),\;x \in X}.
\]
For an integer $k$ it holds that $\sigma \leq k$ if and only if
\[\num\label{dual-program-technique-dual-decision}
    \max\set{g(x)|f(x)\leq k,\;G(x),\;x\in X} \geq b.
\]
Hence, if we can bound $\sigma$ to an interval $[\ell,u]$ and have an algorithm to decide \eqref{dual-program-technique-dual-decision} in at most time $T$ for each $k\in[\ell,u]$, then we can compute $\sigma$ in time $\Oh(T\log(u-\ell))$ by using a binary search for the smallest feasible $k$ in $[\ell,u]$.
At first glance this does not help a lot, as we have replaced one optimization problem by another. However, if $f$ is a function of low complexity, then in order to decide \eqref{dual-program-technique-dual-decision} we can try to utilize an algorithm for \(\max\set{g(x)|G(x),x\in X}\) with respect to the simple additional constraint \(f(x)\leq\kappa\) (which is of low complexity). Remark that the approach may be adapted to approximate $\sigma$ for a function $f:X\rightarrow\R$ as well. However, in this paper we only care about integer-valued objectives. We apply this technique in \cref{response-time-computation,application-to-4-block}.

\subsection{Mixing Set}

Although \textsc{Mixing Set} can be solved in more general settings, for the scope of this paper we will always assume that $a\in\Z_{\geq 1}^n$ and $b\in\Z^n$. In more detail, given capacities $a\in\Z_{\geq 1}^n$, a right-hand side $b\in\Z^n$, and weights $w\in\Z_{\geq 0}^n$ (and $w_0\in\Z_{\geq 0}$) we consider the following minimization problem:
\[\label{mixing-set}
    \min\set{w_0s+\textstyle\sum_{i\leq n} w_ix_i | s+a_ix_i \geq b_i\;\forall i \in [n],\;  s\in\Z_{\geq 0},\; x\in\Z^n}
    \tag{\textsc{Mix}}
\]
Formally, this is not a \textsc{Mixing Set} problem, since $s\in\Z_{\geq 0}$ is required to be an integer instead of a real number. However, due to the integrality of all other numbers the relaxation to $s\in\R_{\geq 0}$ admits just the same set of \emph{optimal} solutions. From \cite{DBLP:journals/orl/ConfortiZ09} (or \cite{DBLP:journals/mp/ZhaoF08,DBLP:conf/ipco/ConfortiSW08}) we can derive the following theorem:
\begin{theorem}[\cite{DBLP:journals/orl/ConfortiZ09}]\label{harmonic-mixing-set}
    For harmonic capacities, \eqref{mixing-set} can be computed in time $\Oh(n^2)$.
\end{theorem}
Given a value for $s$ it is obvious that $x_i = \ceil{(b_i-s)/a_i}$ for all $i\leq n$ leads to an optimal completion of $s$ to a solution $(s,x)$ to \eqref{mixing-set}. The solution $(s,x)$ might not be optimal overall but the choice of $x$ is optimal with respect to $s$. Therefore, given a solution $(s,x)$ we will assume that $x_i = \ceil{(b_i-s)/a_i}$. Furthermore, we will write $x_i(s) = \ceil{(b_i-s)/a_i}$ (and $x(s) = (x_1(s),\dots,x_n(s))$) to denote the dependency of variable $x_i$ on $s$. Remark that by adding/subtracting the least common multiple $m = \lcm_{j\leq n} a_j$ to/from $s$ the value $x_i(s)$ is shifted by $m/a_i$ in the sense that
\[
    x_i(s\pm m)
    = \ceil[\Big]{\frac{b_i-s\mp m}{a_i}}
    = \ceil[\Big]{\frac{b_i-s}{a_i}} \mp \frac{m}{a_i}
    = x_i(s) \mp \frac{m}{a_i}
\]
for all $i\leq n$.
We will depend on an upper bound on the value of $s$ in optimal solutions to \eqref{mixing-set}. Therefore, the following observations will be important.
\begin{observation}\label{mixing-set-unbounded}
    If $\sum_{i\leq n} \frac{w_i}{a_i} > w_0$, then \eqref{mixing-set} is unbounded.
\end{observation}
\begin{proof}
    Let $m=\lcm_{i\leq n} a_i$. For any feasible solution $(s,x)$ we get another feasible solution by adding a multiple of $m$ to $s$. In particular it holds that
    \begin{align*}
        w_0(s+m) + \sum_{i\leq n} w_ix_i(s+m)
        &\;=\; w_0(s+m) + \sum_{i\leq n} w_ix_i(s) - \sum_{i\leq n} w_i\frac{m}{a_i}\\
        &\;=\; w_0 s + \sum_{i\leq n}w_ix_i(s) + m\Big(w_0 - \sum_{i\leq n} \frac{w_i}{a_i}\Big)
    \end{align*}
    and that means that solution $(s',x(s'))$ with $s' = s+m > s$ will be a truly better solution than $(s,x)$ if $\sum_{i\leq n} \frac{w_i}{a_i} > w_0$.
\end{proof}
Hence, we assume that $\sum_{i\leq n}\frac{w_i}{a_i} \leq w_0$ in general.
\begin{observation}\label{mixing-set-small-solution}
    There is an optimal solution to \eqref{mixing-set} with $s < \lcm_{i\leq n} a_i$.
\end{observation}
\begin{proof}
    Let $m=\lcm_{i\leq n} a_i$ and let $(s,x)$ be an optimal solution with $s\geq m$. We find
    \[
        w_0(s-m) + \!\sum_{i\leq n}w_ix_i(s-m)
        = w_0 s + \!\sum_{i\leq n}w_ix_i(s)-m\Big(\!\underbrace{w_0-\!\!\sum_{i\leq n}\frac{w_i}{a_i}}_{\geq\,0}\!\Big)
        \leq w_0 s + \!\sum_{i\leq n}w_ix_i(s).
    \]
    Hence, $(s' = s -m,x'=x(s'))$ is an optimal solution too. Iterate this argument to find an optimal solution with $s$ smaller than $m$.
\end{proof}
By using an analog contradiction proof it is easy to see that \emph{any} optimal solution will hold $s<\lcm_{i\leq n} a_i$ if $\sum_{i\leq n}\frac{w_i}{a_i}$ is truly smaller than $w_0$. We will use this strengthening:
\begin{observation}\label{mixing-set-small-utilization}
    If $\sum_{i\leq n}\frac{w_i}{a_i} < w_0$, then any optimal solution to \eqref{mixing-set} holds $s < \lcm_{i\leq n} a_i$.
\end{observation}
In \cref{lcm-bound-for-S-is-tight} we show that this bound on $s$ is tight in general.

\subsection{Response Time Bounds}
\label{sec:response-time-bounds}

First of all, we refine the utilization bound $\sum_{i<n}c_i/p_i < 1$. In general it holds that
\[\num\label{general-utilization-bound}
    \sum_{i<n}\frac{c_i}{p_i} \leq 1-\frac{1}{\lcm_{i<n}p_i}.
\]
To see this, let $m = \lcm_{j<n}p_j$. From $\sum_{i<n}c_i/p_i < 1$ it follows that $m\cdot\sum_{i<n}c_i/p_i < m$ and both sides of this inequality are integers which implies that $m\cdot\sum_{i<n}c_i/p_i \leq m-1$ and division by $m$ yields the claim.

For the sake of completeness we give some simple bounds to worst-case response times for both the presence and absence of task release jitters. We start with the former case. We also refer to the recent results \cite{DBLP:conf/rtss/BiniPD15,DBLP:conf/rtss/ChenHL16,DBLP:journals/rts/GrassN18,DBLP:conf/ecrts/BoyerRDP21} which give better/comparable bounds. 

\begin{lemma}\label{wcrt-bounds}
    It holds that $\ell \leq r_n \leq u = \min\set{u_1,u_2}$ where
    \begin{align*}
        &\ell = \frac{c_n+\sum_{i<n}\frac{\jitter_i}{p_i}c_i}{1-\sum_{i<n}\frac{c_i}{p_i}},
        \;
        u_1=\ell + \frac{\sum_{i<n}c_i}{1-\sum_{i<n}\frac{c_i}{p_i}},
        \;
        u_2 = \ceil*{\frac{\sum_{i\leq n}c_i}{(1-\sum_{i<n}\frac{c_i}{p_i})\lcm_{i<n}p_i}}\lcm_{i<n}p_i.
    \end{align*}
\end{lemma}
\begin{proof}
    By dropping the roundings it is easy to see that
    \[
        r_n
        \;=\; c_n + \sum_{i<n} c_i\ceil*{\frac{r_n+\jitter_i}{p_i}}
        \;\geq\; c_n + r_n\sum_{i<n} \frac{c_i}{p_i} + \sum_{i<n} \frac{\jitter_i}{p_i}c_i
    \]
    which directly implies the lower bound $\ell$.
    To prove the former upper bound $u_1$ we see that
    \[
        r_n
        \;=\; c_n + \sum_{i<n}\ceil[\Big]{\frac{r_n+\jitter_i}{p_i}}c_i
        \;\leq\; c_n + \sum_{i<n}\Big(\frac{r_n+\jitter_i}{p_i}+1\Big)c_i
    \]
    and rearrangement yields
    \[
        r_n
        \;\leq\; \ell + \frac{\sum_{i<n}c_i}{1-\sum_{i<n}\frac{c_i}{p_i}}
        \;=\; u_1.
    \]
    To confirm the latter upper bound $u_2$ we can argue that
    $u_2 \geq (\sum_{i\leq n}c_i)/(1-\sum_{i<n}c_i/p_i)$ which implies that $\sum_{i\leq n}c_i \leq u_2(1-\sum_{i<n}c_i/p_i)$ and we observe
    \begin{align*}
        c_n+\sum_{i<n}\ceil[\Big]{\frac{u_2+\jitter_i}{p_i}}c_i\;
        &=\; c_n+\sum_{i<n}\Big(\frac{u_2}{p_i}+\ceil[\Big]{\frac{\jitter_i}{p_i}}\Big)c_i\reason{$\lcm$ property}\\
        &\leq\; c_n+\sum_{i<n}\Big(\frac{u_2}{p_i}+1\Big)c_i\reason{$0 \leq \jitter_i \leq p_i$}\\
        &=\; u_2\sum_{i<n}\frac{c_i}{p_i}+\sum_{i\leq n} c_i\\
        &\leq\; u_2\sum_{i<n}\frac{c_i}{p_i}+u_2\Big(1-\sum_{i<n}\frac{c_i}{p_i}\Big)
        \;=\; u_2
    \end{align*}
    which proves the claim. \todo{we never use $u_2$ ... drop it?}
\end{proof}

In the following we will bound the size of the interval $[\ell,u]$.

\begin{lemma}\label{pseudo-polynomial-interval-length}
    By using only the general utilization bound it follows that $u-\ell \leq p_{\max}^n$.
    With the schedulability utilization bound it follows that $u-\ell \leq p_{\max}^2$ and even $u \leq 2p_{\max}^2$.
\end{lemma}
\begin{proof}
    Using the general utilization bound \eqref{general-utilization-bound}, i.e. $\sum_{i<n}c_i/p_i \leq 1-1/\lcm_{i<n}p_i$, it follows that $\sum_{i<n}c_i \leq p_{\max}\sum_{i<n}c_i/p_i < p_{\max}$ and
    \[
        u_1 - \ell
        = \frac{\sum_{i<n}c_i}{1-\sum_{i<n}\frac{c_i}{p_i}}
        \leq (\lcm_{i<n}p_i)\sum_{i<n}c_i
        \leq (\lcm_{i<n}p_i)p_{\max}
        \leq p_{\max}^n.
    \]
    Assume that $\sum_{i\leq n}c_i/p_i \leq 1$ which implies that $1/(1-\sum_{i<n}c_i/p_i) \leq p_n/c_n$. Since $\sum_{i\leq n}c_i \leq p_{\max}\sum_{i\leq n}c_i/p_i \leq p_{\max}$ we have
    \[
        u_1 - \ell
        = \frac{\sum_{i<n}c_i}{1-\sum_{i<n}\frac{c_i}{p_i}}
        \leq \frac{p_n}{c_n}\sum_{i<n}c_i
        \leq \frac{p_n}{c_n}p_{\max}
        \leq p_{\max}^2
    \]
    and by also using the definition of $\ell$ it follows that
    \[
        u_1
        = \frac{\sum_{i\leq n}c_i+\sum_{i<n}\frac{\jitter_i}{p_i}c_i}{1-\sum_{i<n}\frac{c_i}{p_i}}
        \leq \frac{p_n}{c_n}\Big(p_{\max}+\sum_{i<n}c_i\Big)
        \leq \frac{p_n}{c_n}\cdot 2p_{\max}
        \leq 2p_{\max}^2
    \]
    and that proves the claim.
\end{proof}

In the case of jitter-free systems, i.e. $\jitter = 0$, the bounds simplify a lot.
Then the tight upper bound is $P=\lcm_{i\leq n}p_i$; in fact, the upper bound of $c_n \cdot P$, which was given by Bonifaci et al. \cite[Lemma 1]{DBLP:conf/rtss/BonifaciMMW13}, can be improved to $P$. The proof is short:

\begin{lemma}\label{wcrt-bounds-without-jitters}
    If $\jitter=0$, then $c_n/(1-\sum_{i<n}\frac{c_i}{p_i}) \leq r_n \leq P$.
\end{lemma}
\begin{proof}
    The lower bound is implied by \cref{wcrt-bounds}. For the upper bound see that
    \[
        c_n+\sum_{i<n}\ceil*{\frac{P}{p_i}}c_i
        = c_n + P\sum_{i<n}\frac{c_i}{p_i}
        = p_n\frac{c_n}{p_n}+P\sum_{i<n}\frac{c_i}{p_i}
        \leq P\sum_{i=1}^{n}\frac{c_i}{p_i}
        \leq P
    \]
    which proves the claim.
\end{proof}

\section{Response Time Computation}\label{response-time-computation}

We introduce generalized problem formulations as follows. For an integer $\gamma\geq 1$ we aim to compute the general response time
\[\num\label{wcrt-with-jitters}
    \response(\gamma) = \min\set{t|t \geq \gamma + \textstyle\sum_{i<n}\displaystyle c_i\ceil[\Big]{\frac{t+\jitter_i}{p_i}},\; t \in \Z_{\geq 0}}
\]
or the response time for the simplified instance without release jitters
\[\num\label{wcrt-without-jitters}
    \response_0(\gamma) = \min\set{t|t \geq \gamma + \textstyle\sum_{i<n}\displaystyle c_i\ceil[\Big]{\frac{t}{p_i}},\; t \in \Z_{\geq 0}}.
\]
Obviously, $r_n = \response(c_n)$ is the worst-case response time of task $\tau_n$ and $r_n' = \response_0(c_n)$ is the worst-case response time for the simplified instance without release jitters.

\subsection{A Conditional Karp Reduction}
\label{conditional-karp-reduction}

In the following we apply the dualization technique of \cref{dualision} to RTC.
Apparently, \eqref{wcrt-with-jitters} may also be stated using the following formulation as an integer linear program:
\[\label{wcrt-decision}
    \response(\gamma) = \min\set{t|t \geq \gamma + \textstyle\sum_{i<n} c_ix_i,\; p_ix_i \geq t+\jitter_i \;\forall i<n,\; t \in \Z_{\geq 0},\; x \in \Z^{n-1}}
    \tag{$\response$}
\]
Let $u$ be the upper bound of \cref{wcrt-bounds} and let $k \in [u]$. One can readily verify that the associated decision problem $\response(\gamma) \leq_? k$ may be decided by answering $\responsedual(k) \geq_? \gamma$ where
\[\label{wcrt-dual}
    \responsedual(k) = \max\set{t-\textstyle\sum_{i<n}c_ix_i|t \leq k,\; p_ix_i \geq t+\jitter_i\;\forall i<n,\; t\in\Z_{\geq 0},\; x\in\Z^{n-1}}.
    \tag{$\responsedual$}
\]
Let $S$ be the smallest upper bound on the value of $s$ in any optimal solution to the following instance of \textsc{Mixing Set}:
\[\label{wcrt-mixing-set}
    \mixing(k) = \min\set{s+\textstyle\sum_{i<n}c_ix_i | s+p_ix_i \geq k+\jitter_i\;\forall i<n,\;  s\in\Z_{\geq 0},\; x\in\Z^{n-1}}
    \tag{$\mixing$}
\]
We will prove that, if $k \geq S$ holds true, then $    \responsedual(k) = k-\mixing(k)$.
Hence, if $k$ is large enough, i.e. $k\geq S$, we can decide $\response(\gamma) \leq_? k$ by using the following equivalence:
\[\num
    \response(\gamma) \leq k
    \;\;\Leftrightarrow\;\; \responsedual(k) \geq \gamma\notag
    \;\;\Leftrightarrow\;\; k-\mixing(k) \geq \gamma\notag
    \;\;\Leftrightarrow\;\; \mixing(k) \leq k-\gamma \label{response-mixing-equivalence}
\]
Thus, we get the following first result.
\begin{theorem}\label{thm:reduction-to-mixing-set}
    Let $\gamma\in\Z_{\geq 1}$. If $k \geq S$, then there is a reduction from $\response(\gamma) \leq_? k$ to $\mixing(k) \leq_? k-\gamma$ in constant time.
\end{theorem}
In the case of harmonic periods we get harmonic capacities in the instances of \textsc{Mixing Set} which implies an efficient decision algorithm. However, without further insights we can decide $\response(\gamma) \leq_? k$ only for large values of $k$.

If $k \geq S$ holds, remark that optimal solutions to \eqref{wcrt-mixing-set} also do hold $x \in \Z_{\geq 0}^{n-1}$, since $s \leq S \leq k$ which implies $x_i = \ceil{(k-s+\jitter_i)/p_i} \geq \ceil{\jitter_i/p_i} \geq 0$ for all $i<n$. By using \cref{mixing-set-small-utilization} it holds that $S \leq \lcm_{i<n}p_i$, since $\sum_{i<n} \frac{c_i}{p_i} \leq \sum_{i<n} \frac{c_i}{p_i} < 1$.

\begin{lemma}\label{reduction-to-mixing-set}
    If $k \geq S$, then $\responsedual(k) = k-\mixing(k)$.
\end{lemma}
\begin{proof}
    Let $(t,x)$ be an optimal solution to \eqref{wcrt-dual}
    with objective value $k-\sigma$ for some value $\sigma$. We know that $0 \leq t \leq k$. Then by setting $s = k-t \geq 0$ we see that $(s,x)$ is a solution to \eqref{wcrt-mixing-set} since $s+p_ix_i = k-t+p_ix_i \geq k-t+t+\jitter_i = k+\jitter_i$. Its objective value is
    \[
        s+\sum_{i<n} c_ix_i
        = k - t + \sum_{i<n} c_ix_i
        = k-\Big(\!\underbrace{t-\sum_{i<n} c_ix_i}_{=\,k-\sigma}\!\Big)
        = \sigma.
    \]
    Vice-versa, let $(s,x)$ be an optimal solution to \eqref{wcrt-mixing-set} with objective value $\sigma$. We know that $0 \leq s \leq S \leq k$. Hence, by setting $t = k-s \geq 0$ we see that $(t,x)$ is a solution to \eqref{wcrt-dual} since $t+\jitter_i = k-s+\jitter_i = k+\jitter_i-s \leq s+p_ix_i -s = p_ix_i$ 
    and $t = k-s \leq k$. Its objective value is
    \[
        t-\sum_{i<n} c_ix_i
        = k-s - \sum_{i<n} c_ix_i
        = k-\Big(s+\sum_{i<n} c_ix_i\Big)
        = k - \sigma
    \]
    and that proves the claim.
\end{proof}

The power of \cref{reduction-to-mixing-set} crucially depends on the size of $S$. Fortunately, for the concrete formulation \eqref{wcrt-mixing-set} we can do better than \cref{mixing-set-small-utilization}. In fact, $S$ can be bounded by the worst-case response time of the simplified instance without release jitters, i.e. $\response_0(\gamma)$; the following Lemma will be of great benefit, especially for the case of harmonic periods.

\begin{lemma}\label{bound-S-with-jitter-free-response-time}
    It is true that $S \leq \response_0(\gamma)-\gamma$.
\end{lemma}
\begin{proof}
    Let $S' = \response_0(\gamma)-\gamma+1$, define $f(s) = s+\sum_{i<n}c_ix_i(s)$ as the objective function of $\mixing(k)$, and let $s \geq S'$ be a 
    solution to $\mixing(k)$. By the definition of $\response_0(\gamma)$, the response time for the simplified instance without release jitters \eqref{wcrt-without-jitters}, we have
    \[\num\label{bound-S-with-jitter-free-response-time-helper-1}
        \sum_{i<n}c_i\ceil[\Big]{\frac{\response_0(\gamma)}{p_i}} \leq \response_0(\gamma) - \gamma
    \]
    which we use together with $\gamma \geq 1$ to show that
    \[\num\label{bound-S-with-jitter-free-response-time-helper-2}
        \sum_{i<n}c_i\ceil[\Big]{\frac{S'}{p_i}}
        = \sum_{i<n}c_i\ceil[\Big]{\frac{\response_0(\gamma)-\gamma+1}{p_i}}
        \leq \sum_{i<n}c_i\ceil[\Big]{\frac{\response_0(\gamma)}{p_i}}
        \stackrel{\textcolor{gray}{\eqref{bound-S-with-jitter-free-response-time-helper-1}}}{\leq} \response_0(\gamma)-\gamma
        < S'
    \]
    and finally, we can conclude that
    \begin{align*}
        f(s-S')
        &= s-S' + \sum_{i<n}c_i\ceil[\Big]{\frac{k+\jitter_i-s+S'}{p_i}}\\
        &\leq s - S' + \sum_{i<n}c_i\Big(\ceil[\Big]{\frac{k+\jitter_i-s}{p_i}}+\ceil[\Big]{\frac{S'}{p_i}}\Big)\\
        &= f(s) - S' + \underbrace{\sum_{i<n}c_i\ceil[\Big]{\frac{S'}{p_i}}}_{<\,S'}
        \stackrel{\textcolor{gray}{\eqref{bound-S-with-jitter-free-response-time-helper-2}}}{<} f(s)
    \end{align*}
    which reveals $s' = s-S'$ as a truly better solution to $\mixing(k)$. Thus, any optimal solution has to be truly smaller than $S'$ which proves the claim.
\end{proof}

Before presenting our algorithms we want to mention a last ingredient, which simply says that task release jitters make worst-case response times worse. The argument is simple:
\[
    \gamma + \sum_{i<n}c_i\ceil[\Big]{\frac{\response(\gamma)}{p_i}}
    \leq c_n + \sum_{i<n}c_i\ceil[\Big]{\frac{\response(\gamma)+\jitter_i}{p_i}}
    \leq \response(\gamma)
\]
The former inequality is trivial as we only increase the numerators and the latter is the defining inequality of the worst-case response time $\response(\gamma)$. Together they say that $\response(\gamma)$ is a solution for the jitter-free case as well. Thus, we find that
\[\num\label{jitters-make-it-worse}
    \response_0(\gamma) \leq \response(\gamma).
\]

\subsection{Harmonic Periods}\label{wcrt-harmonic-periods}

Here we assume harmonic periods, i.e. $p_i\geq p_j$ implies $p_i/p_j \in \Z$ for all $i,j \leq n$.
At this point we are already suited with all the necessary tools to state our algorithm for harmonic periods.
Together, they allow for the following simple algorithm\footnote{A previous version of this paper presented a much more involved approach to prove polynomiality for harmonic periods; in fact, it is \cref{bound-S-with-jitter-free-response-time} that allows us to achieve this much simpler \emph{and} faster algorithm}.

\begin{algorithm}[H]
\setstretch{1.3}
\begin{algorithmic}[1]
    \State Compute $\response_0(c_n)$, the response time for the simplified instance without release jitters
    \State $L \gets \response_0(c_n)$
    \State $R \gets u$
    \While{$L \neq R$}
        \State $k \gets \floor{(L+R)/2}$
        \IfThenElse{$\mixing(k) \leq k - c_n$}{$R \gets k$}{$L \gets k+1$}
    \EndWhile
    \State \Return $R$
\end{algorithmic}
\caption{\textsc{Harmony}}
\label{alg:harmony}
\end{algorithm}

Now we are ready to give the proof to \cref{main-result}.

\begin{proof}[Proof of \cref{main-result}]
    The first step of \cref{alg:harmony} is to compute $\response_0(c_n)$ and as the periods are harmonic this can be done in at most $\Oh(n\log(n+p_{\max}))$ operations by applying Bonifaci et al. \cite{DBLP:conf/rtss/BonifaciMMW13} or Nguyen et al. \cite{DBLP:conf/iwoca/NguyenGJ22}. The choice of $\response_0(c_n)$ as the lower bound $L$ is reasonable due to \eqref{jitters-make-it-worse}. In each iteration of the subsequent binary search we know that
    \begin{align*}
        k
        &\geq L\reason{Def. $k$}\\
        &= \response_0(c_n)\reason{Def. $L$}\\
        &\geq \response_0(c_n) - c_n\reason{$c_n\geq 1$}\\
        &\geq S\reason{\cref{bound-S-with-jitter-free-response-time}}
    \end{align*}
    which sets us in the position to apply \cref{thm:reduction-to-mixing-set}, i.e. to use the equivalence \eqref{response-mixing-equivalence} to decide whether $r_n = \response(c_n) \leq k$ by solving an appropriate instance of \textsc{Mixing Set} which we do in line 6 of \cref{alg:harmony}. Each instance can be solved in time $\Oh(n^2)$ by using \cref{harmonic-mixing-set}.
    In total we get a running time of at most
    \begin{align*}
        &\Oh(n\log(n+p_{\max})) + \Oh(n^2\log(u-\response_0(c_n)))\\
        &\leq \Oh(n\log(n+p_{\max})) + \Oh(n^2\log p_{\max})\reason{\cref{pseudo-polynomial-interval-length}}\\
        &\leq \Oh(n^2\log p_{\max})
    \end{align*}
    which ends the proof.
\end{proof}

\subsection{Arbitrary Periods}
\label{wcrt-arbitrary-periods}

The following theorem reveals a first simple algorithm to compute worst-case response times for general periods in exponential time.

\begin{theorem}\label{wcrt-simple-approach}
    The worst-case response time $r_n$ can be computed in time $\Oh(n\cdot\lcm_{i<n}p_i)$.
\end{theorem}
\begin{proof}
    Let $m=\lcm_{i<n}p_i$ and $w(t) = c_n+\sum_{i<n}c_i\ceil{\frac{t+\jitter_i}{p_i}}$. Then there is a remainder $\rho \in \set{0,\dots,m-1}$ and some $\lambda \in \Z_{\geq 0}$ such that $r_n = \rho+\lambda m$. We find that
    \[
        \rho+\lambda m
        = r_n
        = w(r_n)
        = w(\rho+\lambda m)
        = w(\rho)+\lambda m\cdot\sum_{i<n}\frac{c_i}{p_i}
    \]
    and thus, $\lambda = (w(\rho)-\rho)/((1-\sum_{i<n}\frac{c_i}{p_i})\cdot m)$ can be computed in time $\Oh(n)$ for a given $\rho$. Hence, we can compute $r_n$ by simply trying each $\rho$ in $\set{0,\dots,m-1}$ and compute $\lambda$ for this choice to find the smallest sum $\rho+\lambda m$ for which the computed $\lambda$ is a non-negative integer. This yields a total running time of $\Oh(n\cdot m)$.
\end{proof}

However, by using our reduction we can prove a Turing reduction as follows.

\begin{theorem}\label{solve-general-wcrt-with-mixing-set}
    The response time $r_n$ can be computed in time $\Oh(\max(Sn,T_{\mixing}\log p_{\max})) \leq \Oh(Sn\log p_{\max})$ where $S$ is a global upper bound on variable $s$ in any optimal solution to \eqref{wcrt-mixing-set} and $T_{\mixing}$ is a global upper bound on the time required to compute \eqref{wcrt-mixing-set}.
\end{theorem}
\begin{proof}[Proof of \cref{solve-general-wcrt-with-mixing-set}]
    We may run an algorithm as follows. First decide $\response(\gamma) \leq_? S$. If the answer is \textsc{yes}, then try out all candidates $t\in[S-1]$ in time $\Oh(Sn)$, i.e. we check whether $t$ is a solution to \eqref{wcrt-with-jitters-original}.
    If this reveals a feasible solution, return the smallest. Otherwise $r_n = S$. If the answer is \textsc{no}, then it holds that $r_n > S$.
    From \cref{pseudo-polynomial-interval-length} we know that $u - \ell \leq p_{\max}^2$ and thus, a subsequent binary search in $[\ell,u]$ leads to a total running time of $\Oh(\max(Sn,T_{\mixing}\log(u-\ell))) \leq \Oh(\max(Sn,T_{\mixing}\log p_{\max}))$. For example we can bound $T_{\mixing}$ as follows.
    We can compute $\mixing(k)$ in time $\Oh(Sn)$ by simply trying solution $(s,x(s))$ for each $s\in [0,S)$. Hence, $T_{\mixing} \leq \Oh(Sn)$ and this gives a total running time of $\Oh(Sn\log p_{\max})$.
\end{proof}

In the case without task release jitter, i.e. $\jitter_i = 0$ for all $i<n$, we can even prove an \emph{unconditional} Karp reduction as follows. For $\jitter = 0$ we get the following form of \eqref{wcrt-mixing-set}:
\[\label{wcrt-mixing-set-without-jitter}
    \mixing_{\jitter=0}(k) = \min\,\{\,s+\textstyle\sum_{i<n}c_ix_i\;|\;
     s+p_ix_i \geq k\;\forall i<n,\;
     s\in\Z_{\geq 0},\; x\in\Z^{n-1}\,\}
    \tag{$\mixing_{\jitter=0}$}
\]
In this case we can show that $S \leq k$. Apparently, $(s=k,x=0)$ is always a solution to \eqref{wcrt-mixing-set-without-jitter} since $s+p_ix_i = k$. It has the objective value $s+\sum_{i<n}c_ix_i = k$. Any solution $(s,x)$ to \eqref{wcrt-mixing-set-without-jitter} with $s>k$ has an objective value truly larger than $k$, since
\[
    s+\sum_{i<n}c_ix_i(s)
    \;=\; s+\sum_{i<n}c_i\ceil*{\frac{k-s}{p_i}}
    \;\geq\; s+\underbrace{(k-s)}_{<\,0}\cdot\underbrace{\sum_{i<n}\frac{c_i}{p_i}}_{<\,1}
    \;>\; s+(k-s)\cdot 1
    \;=\; k.
\]
Hence, any optimal solution $(s,x)$ to \eqref{wcrt-mixing-set-without-jitter} holds $s \leq k$ which implies that $S \leq k$.

Therefore, $\responsedual(k) = k-\mixing(k)$ is directly implied. Thus, we can use a binary search to compute $r_n$ in time $\Oh(T_{\mixing}\log p_{\max})$ where $T_{\mixing}$ is a global upper bound on the time required to compute $\mixing(k)$. This yields the following result.

\begin{theorem}
    If $\jitter=0$, then $r_n$ can be computed in time $\Oh(T_{\mixing}\log p_{\max})$ where $T_{\mixing}$ is a global upper bound on the time required to compute \eqref{wcrt-mixing-set}.
\end{theorem}

\paragraph*{Small Utilization}
Similar to \ac{RM} scheduling (cf. \cite{DBLP:journals/jacm/LiuL73,DBLP:conf/afips/Serlin72}) the complexity of RTC decreases with a smaller utilization  also for \ac{FP} scheduling. Let $U = \sum_{i<n}c_i/p_i$ denote the utilization.

\begin{lemma}\label{bound-S-with-utilization}
    In general it holds that $S < (\sum_{i<n}c_i)/(1-U)$.
\end{lemma}
\begin{proof}
    Let $S' = (\sum_{i< n}c_i)/(1-U)$, define $f(s) = s+\sum_{i<n}c_ix_i(s)$ as the objective function of $\mixing(k)$, and assume that $s \geq S'$ is a solution to $\mixing(k)$. Then
    \begin{align*}
        f(s-S')
        &= s-S' + \sum_{i<n}c_i\ceil[\Big]{\frac{k+\jitter_i-s+S'}{p_i}}\\
        &\leq s - S' + \sum_{i<n}c_i\Big(\ceil[\Big]{\frac{k+\jitter_i-s}{p_i}}+\ceil[\Big]{\frac{S'}{p_i}}\Big)\\
        &= f(s) - S' + \sum_{i<n}c_i\ceil[\Big]{\frac{S'}{p_i}}
        < f(s) - S' + \sum_{i<n}c_i\Big(\frac{S'}{p_i}+1\Big)
        = f(s)
    \end{align*}
    reveals $s' = s-S'$ as a truly better solution to $\mixing(k)$.
\end{proof}
With the schedulability utilization bound this immediately assures that $S$ is pseudo-polynomial and using the general utilization bound only it suffices for example that $U \leq 1-p_{\max}^{-\Oh(1)}$. Since $u_1-\ell = (\sum_{i<n}c_i)/(1-U)$ this also implies a pseudo-polynomial algorithm by simply guessing the right value in the interval $[\ell,u]$. However, especially if $r_n > S$ our algorithm may be faster.

\section{Mixing Set}
\label{sec:mixing-set}

Our technique can be reversed to solve the \textsc{Mixing Set} problem by computing worst-case response times. This allows us to make efficient approaches to RTC applicable to solve \textsc{Mixing Set}. We simply do a substitution in \eqref{response-mixing-equivalence} (substitute $k \mapsto k'+\gamma$, $\gamma \mapsto \beta - k'$, $k' \mapsto k$ in this order) to achieve the following corollary.

\begin{corollary}\label{reverse-reduction}
    Let $k \geq 0$, and $\beta \geq S$. Then $\mixing(\beta) \leq k$ if and only if $\response(\beta-k) \leq \beta$.
\end{corollary}

With the reverse reduction we can improve the running time for general capacities, if the entries of the right-hand side $b$ are large and close to $b_{\min}$:

\begin{lemma}\label{mixing-set-with-large-crowded-right-hand-side}
    Given capacities $a\in\Z_{\geq 1}^n$, a right-hand side $b\in\Z^n$ which holds that $\lcm_{j\leq n} a_j \leq b_i \leq b_{\min}+a_i$ for all $i\leq n$, and weights $w \in \Z_{\geq 0}^n$, one can compute
    \[\label{mixing-with-crowded-right-hand-side}
        \textstyle
        \min\,\{\,s+\sum_{i\leq n}w_ix_i\;|\;
         s+a_ix_i\geq b_i\;\forall i\leq n,\;
         s\in\Z_{\geq 0},\; x\in\Z^n\,\}\tag{$\mixing_b$}
    \]
    in time $\Oh(T_{\response}\log b_{\max})$ where $T_{\response}$ is a global bound on the time required to compute \eqref{wcrt-decision}.
\end{lemma}
\begin{proof}
    Choose $\beta = b_{\min}$ and $\jitter_i = b_i - b_{\min}$ for all $i \in [n]$. We get $\beta \geq \lcm_{i\leq n} a_i \geq S$ and $0 \leq \jitter_i \leq a_i$ for all $i\leq n$. This implies $b_i = \beta + \jitter_i$ and thus, we can write \eqref{mixing-with-crowded-right-hand-side} as
    \[
        \min\,\{\,s+\textstyle\sum_{i\leq n}w_ix_i\;|\;
         s+a_ix_i\geq \beta+\jitter_i\;\forall i\leq n,\;
         s\in\Z_{\geq 0},\; x\in\Z^n\,\}.
    \]
    Now, we use \cref{reverse-reduction} to start a binary search for the optimal $k$. We know that $\response(\beta-k)$ can be computed in time $T_{\response}$ and the objective value of our \textsc{Mixing Set} instance is upper bounded by $b_{\max}$ (since $(s=b_{\max},x=0)$ is always a solution with objective value $b_{\max}$) which implies a total running time of $\Oh(T_{\response}\log b_{\max})$.
\end{proof}
In fact, \cref{mixing-set-with-large-crowded-right-hand-side} implicitly solves the general case. We use it to prove the following.
\begin{theorem}\label{mixing-set-with-large-crowded-right-hand-side-general}
    Given capacities $a\in\Z_{\geq 1}^n$, a right-hand side $b\in\Z^n$, and weights $w \in \Z_{\geq 0}^n$, one can compute
    \[\label{large-crowded-ip}
        \textstyle
        \min\,\{\,s+\sum_{i\leq n}w_ix_i\;|\;
         s+a_ix_i\geq b_i\;\forall i\leq n,\;
         s\in\Z_{\geq 0},\; x\in\Z^n\,\}\tag{$\mixing_b$}
    \]
    in time $\Oh(n\!+\!T_{\response}\log b_{\max})$ where $T_{\response}$ is a global bound on the time required to compute \eqref{wcrt-decision}.
\end{theorem}
\begin{proof}
    Let $m = \lcm_{j\leq n}a_j$. We define a new right-hand side $b' \in \Z^n$ by
    \(
        b_i' = b_i + \ceil{(m - b_i)/a_i}a_i
    \)
    for each $i\leq n$. It is easy to see that $m \leq b_i' \leq m+a_i \leq b'_{\min}+a_i$ for all $i\leq n$. Also, we have
    \begin{align*}
        &\min\set{s+\textstyle\sum_{i\leq n}w_ix_i|s+a_ix_i \geq b'_i\forall i\leq n, s\in\Z_{\geq 0},x\in\Z^n}\\
        &= \min\set{s+\textstyle\sum_{i\leq n}w_i\ceil*{\frac{b'_i-s}{a_i}}| s\in\Z_{\geq 0}}\\
        &= \textstyle\sum_{i\leq n}w_i\ceil*{\frac{m - b_i}{a_i}} + \min\set{s+\textstyle\sum_{i\leq n}w_i\ceil*{\frac{b_i-s}{a_i}}| s\in\Z_{\geq 0}}\\
        &= \textstyle\sum_{i\leq n}w_i\ceil*{\frac{m - b_i}{a_i}} + \eqref{large-crowded-ip}
    \end{align*}
    where we applied
    \[\ceil*{\frac{b_i'-s}{a_i}} = \ceil*{\frac{b_i+\ceil[\big]{\frac{m-b_i}{a_i}}a_i-s}{a_i}} = \ceil*{\frac{m-b_i}{a_i}}+\ceil*{\frac{b_i-s}{a_i}}.\]
    Hence, we can solve the system for $b'$ with \cref{mixing-set-with-large-crowded-right-hand-side} and subtract $\sum_{i\leq n}w_i\ceil{(m - b_i)/a_i}$ to obtain the optimal value. This gives the desired running time.
\end{proof}

Now, we aim to identify the instances of \textsc{Mixing Set} which may be solved by using our reversed reduction in combination with efficient algorithms for the computation of worst-case response times. By turning to harmonic periods $p_i$ and using \cref{reverse-reduction} again, we can apply Nguyen et al. \cite{DBLP:conf/iwoca/NguyenGJ22} to compute
\[
    \textstyle
    \min\set{s+\sum_{i<n}c_ix_i|s+p_ix_i\geq \beta\;\forall i<n,\; s\in\Z_{\geq 0},\; x\in\Z^{n-1}}
\]
in time $\Oh(n\log n+n\log\beta)$ for any $\beta \geq \max_{i<n} p_i$ by computing the response times
\[
    \min\set{t|t \geq \beta - k + \textstyle\sum_{i<n} \ceil[\big]{\frac{t}{p_i}}c_i,\; t \in \Z_{\geq 0}}
\]
in a binary search for the smallest feasible $k$. It suffices to search in the interval $[0,\beta]$ since $(s=\beta,x=0)$ is always a solution with value $\beta$. Remark that the usual running time to solve instances of \textsc{Mixing Set} with harmonic capacities is $\Oh(n^2)$. In a more classical formulation this yields the following theorem.

\begin{theorem}
    Given harmonic capacities $a\in\Z_{\geq 1}^n$, a right-hand side value $\beta\geq a_{\max}$, and weights $w \in \Z_{\geq 0}^n$, one can compute
    \[\label{mixing-beta}
        \textstyle
        \min\,\{\,s+\sum_{i\leq n}w_ix_i\;|\;
         s+a_ix_i\geq \beta\;\forall i\leq n,\;
         s\in\Z_{\geq 0},\; x\in\Z^n\,\}\tag{$\mixing_{\beta}$}
    \]
    in time $\Oh(n\log n+n\log\beta) \leq \Oh(n\log(n+\beta))$.
\end{theorem}

The algorithm of Nguyen et al. runs in linear time $\Oh(n)$ after an initial sorting step in time $\Oh(n\log n)$ to sort the tasks by their periods. As we seek to apply the algorithm multiple times, we can save the time to sort by only doing it once.

However, for arbitrary capacities and the special case $b_i=\beta$ for all $i\leq n$ and $\beta \geq \lcm_{j\leq n}a_j$ we can also derive the following corollary from \cref{mixing-set-with-large-crowded-right-hand-side}.
\begin{corollary}
    Given capacities $a\in\Z_{\geq 1}^n$, a right-hand side value $\beta\geq\lcm_{i\leq n}a_i$, and weights $w \in \Z_{\geq 0}^n$, one can compute
    \eqref{mixing-beta}
    in time $\Oh(T_{\response}\log\beta)$ where $T_{\response}$ is a global upper bound on the time required to compute \eqref{wcrt-decision}.
\end{corollary}

\section{4-Block Integer Programming}
\label{application-to-4-block}

Block-structured integer programming has received major interest in the recent past. Namely, the three most important lines of research are $n$-fold integer programming, $2$-stage stochastic integer programming, and $4$-block integer programming \cite{DBLP:conf/icalp/EisenbrandHK18,DBLP:conf/esa/0009K0S20,arxiv/2201.05874,DBLP:conf/ipco/HemmeckeKW10}. The last one is the most general but also the least understood flavor of block-structured integer programs. In this section we will consider 4-block integer programs with only one coupling constraint inequality. Furthermore, we consider a simpler objective function which addresses at most two \emph{bricks}. From \cref{4-block-hardness} it follows that this problem is still hard to approximate to any constant factor.
In fact, the literature on $4$-block integer programs is not unified in its formal presentation. Here we consider $4$-block integer programs as follows.

Consider matrices $A_1,\dots,A_n$, $B_1,\dots,B_n$, $C_1,\dots C_n$, and $D$ such that $A_i \in \Z^{r\times t}$, $B_i \in \Z^{r\times s}$, and $C_i\in\Z^{q\times t}$ for all $i\in[n]$ and $D\in\Z^{q\times s}$. We aim to compute
\[\num\label{4-block}
    \min\set{w^\intercal x|\mathcal{B}x = b, \; 0 \leq x \leq u, \; x\in \Z^{s+nt}}
\]
where
\[
    \mathcal{B}=\begin{pmatrix}
        D      & C_1 & \cdots & C_n \\
        B_1    & A_1                \\
        \vdots &     & \ddots       \\
        B_n    &     &        & A_n
    \end{pmatrix}.
\]

We use the notion of \emph{bricks} 
to refer to different areas of $w$ and $x$. For each $i\in\set{1,\dots,n}$ let $w^{(i)} = (w_{(i-1)t+1},\dots,w_{it})$ and $x^{(i)} = (x_{(i-1)t+1},\dots,x_{it})$ denote the $i$-th bricks of $w$ and $x$. Furthermore, let $w^{(0)} = (w_1,\dots,w_s)$ and $x^{(0)} = (x_1,\dots,x_s)$ denote the $0$-th bricks of $w$ and $x$.

In the following we assume that there is exactly one coupling constraint \emph{inequality} and that there is a $j \in [n]$ such that $w^{(i)}=0$ for all $i\in[n]\setminus\set{j}$. In other words, the objective function addresses exactly two bricks of the solution and at least one of them is the first brick.
Let $a \in \Z^{s+nt}$ denote the first row of $\mathcal{B}$. We aim to compute
\[\num\label{simple-4-block}
    \tau=\min\set{w^\intercal x|a^\intercal x \geq b_0,\;\mathcal{B}'x = b, \; 0 \leq x \leq u, \; x\in \Z^{s+nt}}
\]
where
\[
    \mathcal{B}'=\begin{pmatrix}C_1&B_1\\\vdots&&\ddots\\C_n&&&B_n\end{pmatrix}.
\]
For the decision problem $\tau \leq_? k$ we get the following dual variant:
\[
    \max\set{a^\intercal x|w^\intercal x\leq k,\;\mathcal{B}'x = b,\; 0 \leq x \leq u, \; x\in \Z^{s+nt}} \geq_? b_0
\]
We introduce a slack variable $y\geq 0$ to replace inequality \enquote{$w^\intercal x\leq k$} with equation \enquote{$w^\intercal x+y = k$}
and by placing this equation on top of the $j$-th block matrices we get a new matrix
\[
    \mathcal{B}''=\begin{pmatrix}
         \tikz[scale=0.6]{
            \footnotesize
            \draw (0,4) rectangle node {$C_1$} (1,5);
            \draw (1,4) rectangle node {$B_1$} (2,5);
            \draw (0,3) rectangle node {$\overbar{w}^{(0)}$} (1,3.4);
            \draw (4,3) rectangle node {$1$} (4.4,3.4);
            \draw (3,3) rectangle node {$\overbar{w}^{(j)}$} (4,3.4);
            \draw (0,2) rectangle node {$C_j$} (1,3);
            \draw [dashed] (4,2) -- (4.4,2) -- (4.4,3);
            \draw (3,2) rectangle node {$B_j$} (4,3);
            \draw (0,0) rectangle node {$C_n$} (1,1);
            \draw (5,0) rectangle node {$B_n$} (6,1);
         }
    \end{pmatrix}
\]

where $\overbar{w}^{(0)}=(w^{(0)})^\intercal$ and $\overbar{w}^{(j)}=(w^{(j)})^\intercal$. Finally, we arrive at the $2$-stage stochastic integer program
\[\num\label{2-stage}
    \max\set{a^\intercal x|\mathcal{B}''x = b',\; 0 \leq x \leq u, \; x\in \Z^{s+nt}, \; y \in \Z_{\geq 0}}
\]
where $b'$ is the modified right-hand side with value $k$.
Use, for example, \cite{DBLP:conf/esa/CslovjecsekEPVW21} to solve it.

\begin{theorem}
    Program \eqref{simple-4-block} can be solved in time $\Oh(T_{\eqref{2-stage}}\cdot\log(2H))$ where $T_{\eqref{2-stage}}$ is a global upper bound on the time to solve \eqref{2-stage} and $H$ is an upper bound on the abs. objective value of an optimal solution to \eqref{simple-4-block}.
\end{theorem}

Using the state-of-the-art algorithm for 2-stage stochastic integer programming by Cslovjecsek et al. \cite{DBLP:conf/esa/CslovjecsekEPVW21} we can bound $T_{\eqref{2-stage}} \leq 2^{(2\Delta)^{\Oh(s(s+t))}}\cdot\log^{\Oh(st)}n$ by using parallel computation where $\Delta = \max\set{\norm{\mathcal{B}}_{\infty},\norm{w}_{\infty}}$ is the largest absolute value of the entries in $\mathcal{B}$ and $w$.
To give a simple bound on $H$ one can use standard arguments (e.g. \cite[Theorem 13.5]{DBLP:books/ph/PapadimitriouS82}) to show that $\log(2H) \leq \Oh(\log(st\norm{w}_1\norm{b}_{\infty})+(nr+q)\log((nr+q)\norm{A}_{\infty}))$.

\section{Conclusion}

To the best of our knowledge the presented connection between RTC and \textsc{Mixing Set} was unknown before. In the case of harmonic periods we derived a simple and efficient algorithm to solve RTC in polynomial time even in the presence of task release jitters.

There are several avenues for future work. Although our results generalize over existing work (e.g. \cite{DBLP:conf/rtss/BonifaciMMW13,DBLP:conf/iwoca/NguyenGJ22}) there is more space for generalization. For example, the relaxation to arbitrary deadlines (e.g. \cite{DBLP:journals/tc/BiniNRB09,DBLP:conf/rtss/Lehoczky90,DBLP:conf/rtss/BiniPD15,DBLP:journals/rts/GrassN18}), which potentially exceed their task's period, is an exciting direction.

Various authors of previous results where concerned about the investigation of fixed point iterations for which they proved convergence (e.g. \cite{DBLP:journals/iee/AudsleyBRTW93,DBLP:journals/rts/TindellBW94,DBLP:conf/rtss/SjodinH98,DBLP:conf/rtss/BiniPD15}). However, to the best of our knowledge they do not give worst-case bounds for the actual running time, i.e. the number of steps to arrive at the smallest fixed point. Therefore, it remains as an interesting question how our algorithms compete in practice.

Another avenue is to ask for new complexity results. Ekberg \cite{DBLP:conf/rtss/Ekberg20} did not consider task release jitters. Therefore, a natural question is: Where are the utilization borders between polynomial and \NP-hard schedulability testing in the presence of task release jitters?


\bibliographystyle{IEEEtran}
\bibliography{bib}

\clearpage

\appendix

\section{Hardness of Approximation of 4-Block IPs}
\label{4-block-hardness}

Eisenbrand and Rothvoß \cite{DBLP:conf/rtss/EisenbrandR08} proved that worst-case response times cannot be approximated to a constant factor, unless $\op{P}=\NP$, even if the task system is jitter-free and has a total utilization truly smaller than $1$. Here we reduce RTC to a 4-block integer linear program. The worst-case response time of a task system without jitters is
\[
    r_n = \min\set{t|t\geq c_n+\textstyle\sum_{i<n}\ceil[\big]{\frac{t}{p_i}}c_i, \; t\in\Z_{\geq 0}}
\]
which may be restated as an integer linear program:
\[
    r_n = \min\set{t|t \geq c_n+\textstyle\sum_{i<n}x_ic_i, \; p_ix_i \geq t \;\forall i\in[n-1], \; t\in\Z_{\geq 0}, \; x\in\Z^{n-1}}.
\]
Thus, we can introduce a 4-block matrix $\mathcal{A}\in\Z^{n\times n}$ and a right-hand side $b\in\Z_{\geq 0}^n$ leading to
\[
    r_n = \min\set{t|\mathcal{A}\begin{pmatrix}t\\x\end{pmatrix}\geq b, \; t\in\Z_{\geq 0}, \; x\in\Z^{n-1}}
\]
where
\[
    \mathcal{A} = \begin{pmatrix}
        \begin{array}{c|ccccc}
             1     & -c_1 & \cdots & -c_{n-1}\\\hline
            -1     &  p_1 &        &         \\
            \vdots &      & \ddots &         \\
            -1     &      &        &  p_{n-1}
        \end{array}
    \end{pmatrix}
    \quad
    \text{and}
    \quad
    b = \begin{pmatrix}c_n\\\hline0\\\vdots\\0\end{pmatrix}.
\]
As $p_ix_i \geq t \geq 0$ and $p_i\geq 1$ imply that $x_i \geq 0$ this leads to the following formulation as a 4-block integer program:
\[
    r_n = \min\set{x_1|\mathcal{A}x\geq b, \; x \geq 0,\; x\in\Z^n}
\]

Therefore, it is NP-hard to approximate 4-block ILPs to a constant factor, even if all block matrices have size $1\times 1$, all constraints are inequalities, the objective function addresses the first variable only, the variables are unbounded, and the right-hand side has at most one non-zero entry.




\section{The Bound of Observation \ref{mixing-set-small-solution} is Tight for General Instances}
\label{lcm-bound-for-S-is-tight}


Here we prove that for optimal solutions to general instances of \textsc{Mixing Set} the value $\lcm_{i\leq n}a_i-1$ is indeed the best upper bound on $s$.
Given an arbitrary number $n\in\Z_{\geq 2}$ we define an instance of \textsc{Mixing Set} by $w_i = 2^i$, $a_i = n\cdot2^i$, and $b_i = n\cdot2^n-1$ for all $i\in[n]$. Remark that $\lcm_{i\leq n}a_i = a_n = n\cdot2^n$ and $\sum_{i\leq n}w_i/a_i = \sum_{i\leq n}1/n = 1$. We get the objective function
\[
    f(s)
    = s+\sum_{i\leq n}w_i\ceil*{\frac{b_i-s}{a_i}}
    = s+\sum_{i\leq n}2^i\ceil*{\frac{n\cdot2^n-1-s}{n\cdot2^i}}.
\]
First consider $0 \leq s \leq n\cdot2^n-2$. We have $1 \leq n\cdot2^n-1-s \leq n\cdot2^n-1$ and at least for $i=n$ this yields
\[
    \ceil*{\frac{n\cdot2^n-1-s}{n\cdot2^i}} = 1 > \frac{n\cdot2^n-1-s}{n\cdot2^i}
\]
which implies that
\[
    f(s)
    = s+\sum_{i\leq n}2^i\ceil*{\frac{n\cdot2^n-1-s}{n\cdot2^i}}
    > s+\sum_{i\leq n}\frac{n\cdot2^n-1-s}{n}
    = \sum_{i\leq n}\frac{n\cdot2^n-1}n
    = n\cdot2^n-1
\]
and from the integrality of both sides it follows that $f(s) \geq n\cdot2^n$.
However, for $s = n\cdot2^n-1$ we have
\[
    f(s)
    = s+\sum_{i\leq n}2^i\ceil*{\frac{n\cdot2^n-1-s}{n\cdot2^i}}
    = n\cdot2^n-1
    < n\cdot2^n
\]
which reveals $\lcm_{i\leq n}a_i-1$ as the smallest optimal solution.

\end{document}